\pdfoutput=1
\documentclass[12pt]{article}

\usepackage[margin=1in]{geometry}

\usepackage{amsmath,amsthm,amssymb}
\usepackage{bm}
\usepackage{braket}

\usepackage{graphicx}
\usepackage{algorithm}
\usepackage{algpseudocode}

\usepackage[colorlinks=true,linkcolor=blue,citecolor=blue,urlcolor=blue]{hyperref}

\newtheorem{theorem}{Theorem}
\newtheorem{lemma}{Lemma}

\newtheorem{definition}{Definition}

\begin{document}

\title{Quantum Advantage in Learning Mixed Unitary Channels}

\author{
    Yue Tu\thanks{Email: \texttt{tuyue3@gmail.com}.} 
    \and 
    Liang Jiang\thanks{Email: \texttt{liangjiang@uchicago.edu}.}
}

\date{November 4, 2025}

\maketitle

\begin{abstract}
We study the task of learning mixed unitary channels using Fisher information, under different quantum resource assumptions including ancilla and concatenation. Our result shows that the asymptotic sample complexity scales as $\frac{r}{d\varepsilon^2}$, where $r$ is the rank of the channel (i.e.\ the number of different unitaries), $d$ is the dimension of the system, and $\varepsilon^2$ is the mean-square error. Thus the critical resource is the ancilla, which mirrors the result in~\cite{chen2022quantum} but in a more precise form, as we point out that $r$ is also important. Additionally, we demonstrate the practical potential of mixed unitary channels by showing that random mixed unitary channels are easy to learn.
\end{abstract}

\tableofcontents

\section{Introduction}

Since the emergence of the concept of quantum computation, researchers have been seeking tasks where quantum computers could outperform their classical counterparts. Significant progress has been made in using quantum computers to solve computational problems, such as Shor’s algorithm, Grover’s algorithm, and quantum Hamiltonian simulation. More recently, researchers have discovered that quantum computers can also be used for the \emph{learning of quantum systems}, which, in contrast to previous applications, does not utilize quantum computation for processing information but rather for extracting information from quantum systems more efficiently.

This new paradigm for demonstrating quantum advantage is known as \emph{Quantum Learning}. Learning a quantum system is inherently difficult because information can only be obtained through quantum measurements, which yield random outcomes according to certain probability distributions and collapse the system’s state. Consequently, it is inevitable that many samples must be prepared and measured repeatedly. Quantum learning aims to use quantum computational resources to reduce the number of samples required, compared with the classical approach of performing projective measurements directly on the system. Different learning algorithms target different types of quantum systems, employ various parameterizations, and utilize different kinds of quantum resources.

In the study of quantum channel learning, considerable progress has been achieved for several specific models of quantum channels. In particular, one useful modeling framework is the \emph{finite mixture parameterization}, where the channel is assumed to apply one operation from a predefined set with certain probabilities. This approach is more tractable, as it avoids the exponential growth of parameters associated with multi-qubit operations. The most well-known example of this type is the Pauli channel, which on $n$ qubits applies an $n$-qubit Pauli operator $P_{\boldsymbol{a}}$ with probability $p_{\boldsymbol{a}}$, i.e.,
\begin{equation}
    \Lambda_P(\cdot) = \sum_{a\in Z_2^{2n}}p_a P_a(\cdot)P_a.
\end{equation}
Pauli channels have been extensively investigated in \cite{chen2022quantum}, \cite{flammia2020efficient} and \cite{chen2024tight}.

In this work, we study a more general model than the Pauli channel—namely, the \emph{mixed unitary channel}, which applies one unitary operation from a predefined set with a certain probability distribution. Specifically:
\begin{equation}\label{eq:muc}
    \Lambda(\cdot) = \sum_{a=1}^r p_a U_a (\cdot) U_a^\dagger,
\end{equation}
where the unitaries $\{U_a\}_{a=1}^r$ are known and the probabilities $\{p_a\}_{a=1}^r$ are unknown.

Compared with the Pauli channel, the mixed unitary channel possesses two additional degrees of freedom. The first is that the number of unitaries, denoted by $r$ and later referred to as the \emph{rank}, is flexible, whereas for the Pauli channel this number is fixed at $4^n$, where $n$ is the number of qubits. The second is that the operations can be general unitaries, not restricted to Pauli operators. These two degrees of freedom have not been systematically studied before, and we argue that this unexplored regime is both theoretically valuable and practically relevant.

The first degree of freedom---the rank---is of particular theoretical interest, as we will later show that it determines how many ancilla qubits are required to achieve optimal learning efficiency. In contrast to previous work~\cite{chen2022quantum}, which showed that learning a fixed-rank Pauli channel requires at least $n$ ancilla qubits to avoid exponential overhead, our sample complexity lower bound,
\begin{equation}
    N \ge \frac{r}{d\varepsilon^2} = \frac{r}{2^{n+m}\varepsilon^2},
\end{equation}
implies that the necessary number of ancilla qubits satisfies $m = \log_2(r / 2^n)$. Hence, by reducing the rank $r$, we can in principle reduce the number of required ancilla qubits---often considered a critical and costly resource to prepare in practice.

The second degree of freedom highlights the potential application of the mixed unitary channel as an \emph{error mitigation model}. Compared with the Pauli channel, which has been widely adopted as the standard model for error mitigation~\cite{cai2023qem}, the mixed unitary model allows one to select from a broader set of unitaries beyond the Pauli group. This additional flexibility can potentially yield a more accurate representation of realistic noise processes. Interestingly, we also show a counterintuitive result: random unitary channels are, in fact, easier to learn. Thus, the intricate mathematical structure of the Pauli group is not a prerequisite for efficient learnability.

We emphasize that our goal is not to study the most general learning model. Since our main focus is on understanding the role of quantum resources---such as ancilla systems and circuit concatenation---we restrict our attention to the \emph{independent one-measurement-per-state} (IOMS) model. In this setting, there are no mid-circuit measurements, feedforward control, or adaptive strategies. As shown in~\cite{chen2024tight}, these resources are not essential in the sense that they do not provide exponential improvement in efficiency. Our analysis is grounded in the Fisher-information framework (originally suggested by Hyukgun Kwon), which has the advantage of yielding sample-complexity bounds with clear, intuitive interpretations of how ancilla dimensions and channel rank affect learnability. The trade-off, however, is that our guarantees are expressed in terms of the mean-square error (MSE) in the asymptotic regime.

In the following, we will first clarify our notation, then provide a precise definition of the Fisher information framework used in our analysis. Next, we present our sample complexity lower bound for learning mixed unitary channels, and finally, we demonstrate that random mixed unitary channels are indeed easy to learn.

\section{Notation}

We summarize the main notations used throughout the paper.
\begin{itemize}
    \item $\mathcal{H}$: Hilbert space of dimension $d$.
    \item $\rho$: quantum state (density operator) on $\mathcal{H}$.
    \item Mixed unitary channel: $\Lambda(\rho) = \sum_{a=1}^r p_a U_a \rho U_a^\dagger$, where $\{U_a\}$ are known unitaries of dimension $d_\Lambda$ and $p = (p_1, \ldots, p_r)$ is an unknown probability vector.
    \item $n$: number of qubits for the channel sample ($2^n = d_\Lambda$).
    \item $m$: number of ancilla qubits for the channel sample ($2^{n+m} = d$).
    \item $N$: number of channel uses (samples).
    \item $\theta$: parameter vector in $\mathbb{R}^r$.
    \item $I(\theta)$: Fisher information matrix associated with the statistical model $P(X\mid \theta)$.
    \item $\mathbb{E}[\cdot]$: expectation over measurement outcomes.
    \item $\varepsilon$: target root-mean-square estimation error.
    \item $\mathbf{1}$: all-ones vector, i.e.\ $\mathbf{1} = (1,1,\ldots,1)^\top \in \mathbb{R}^r$.
    \item $u$: uniform probability vector, i.e.\ $u = \mathbf{1}/r$.
\end{itemize}
All logarithms are base $2$ unless otherwise stated.

\section{Main Results}

In this section, we present our main results on the sample complexity for the task of learning a mixed unitary channel. A mixed unitary channel applies one unitary operation from a set of $r$ unitaries according to a discrete probability distribution $p$:
\begin{equation}
    \Lambda(\rho) = \sum_{a=1}^r p_a\, U_a \rho U_a^\dagger.
\end{equation}
Our goal is to estimate $p$, assuming that the unitaries $\{U_a\}$ are known. Specifically, we aim to determine the minimal scaling of the number of samples required to learn $p$ up to a mean-square error of $\varepsilon^2$ in the asymptotic regime as $\varepsilon \rightarrow 0$. Each use of the channel is counted as one sample.

Consider the most general learning protocol, in which the learner is allowed to perform any operation permitted by quantum mechanics. Such a protocol may involve four types of operations:  
(1) preparing a quantum state,  
(2) applying a unitary operation,  
(3) applying the mixed unitary channel, and  
(4) performing a POVM measurement.  
Any learning strategy can be represented as a sequence of these operations, which ultimately produces an estimate $\tilde{p}$ based on the collected measurement outcomes.

However, we do not aim to study the most general learning scenario. Instead, we impose two restrictions.  
First, after each measurement is performed on a quantum state, the state is discarded.  
Second, the learner does not adapt future operations based on previous measurement results.  
We refer to this restricted setting as the \emph{independent one-measurement-per-state} (IOMS) protocol.  
In the IOMS setting, each learning round independently performs a fixed sequence of operations:  
it begins with (1) state preparation, then executes (2) unitary operations and (3) channel applications (possibly multiple times), and concludes with (4) measurement.  
This procedure is illustrated in Fig.~\ref{fig:one-measure-proto}.

\begin{figure}[t]
    \centering
    \includegraphics[width=0.5\linewidth]{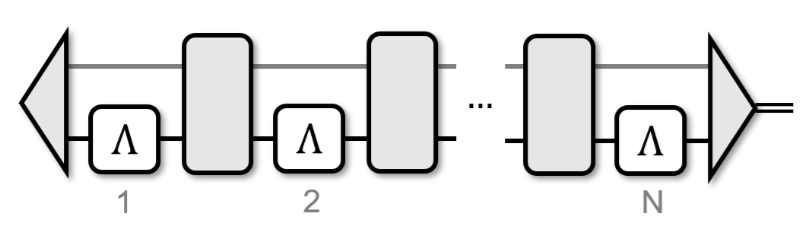}
    \caption{Schematic illustration of the independent one-measurement-per-state (IOMS) protocol.}
    \label{fig:one-measure-proto}
\end{figure}

Although restricted, the IOMS protocol is sufficiently rich to yield meaningful insights into how quantum resources—such as ancilla systems and channel concatenation—contribute to the learning process.  
Here, \emph{ancilla} refers to the ability to prepare an entangled state whose total dimension exceeds that of the quantum channel.  
\emph{Concatenation} refers to the ability to apply the mixed unitary channel multiple times on a single quantum state before measurement.  
Both ancilla and concatenation are considered valuable quantum resources, as they represent capabilities that are only accessible to a quantum computer.

In the following, we present the proofs of our main results, divided into three parts.  
First, we map the general IOMS protocol to a well-defined statistical model on which Fisher information can be defined.  
Second, we establish a sample-complexity lower bound for general IOMS protocols with a given number of ancilla qubits and channel concatenations.  
Finally, we propose an algorithm that efficiently learns random mixed unitary channels.

\subsection{Definition of Fisher Information}

In this section, we map the IOMS protocol to a parametric statistical model on which we can define the Fisher information. A (parametric) statistical model is specified by a random variable $X$ taking values in a sample space $\mathcal{X}$, together with a family of probability distributions
\begin{equation}
    \{ P_\theta(X \in \cdot) : \theta \in \Theta \}
\end{equation}
on $\mathcal{X}$, parameterized by $\theta$.  
For the IOMS protocol, $X$ corresponds to the collection of POVM measurement outcomes, and $\theta$ represents the parameters of the mixed unitary channel, i.e., the probability weights associated with different unitaries.

We now define the Fisher information, which quantifies how sensitive the probability distribution $P_\theta$ is to infinitesimal changes in the parameter $\theta$. By definition,
\begin{equation}\label{eq:FImat-def}
    I(\theta) = \mathbb{E}\!\left[
        (\nabla_\theta \log P_\theta(X)) \, (\nabla_\theta \log P_\theta(X))^\top
    \right].
\end{equation}

In the IOMS setting, $X$ consists of multiple rounds of measurement. Each measurement is performed on a quantum state $\rho^{\theta; X^{(<i)}}$ using a POVM element $E^{i; X^{(<i)}}$.  
Since the strategy across different measurements is independent, we have
\begin{equation}
    \rho^{\theta; X^{(<i)}} = \rho^{\theta}, 
    \qquad 
    E^{i; X^{(<i)}} = E^{i}.
\end{equation}
Because the pair $(\rho^{\theta}, E^{i})$ fully determines the measurement distribution, we obtain
\begin{equation}
    P_\theta(X^{i} \mid X^{(<i)}) 
    = P(X^{i} \mid \rho^{\theta; X^{(<i)}}, E^{i; X^{(<i)}})
    = P(X^{i} \mid \rho^{\theta}, E^{i}).
\end{equation}

It is well known that the total Fisher information decomposes as the sum of the single-sample Fisher information terms corresponding to the conditional distributions:
\begin{equation}
    I(\theta) = \sum_{i=1}^{N} I_i(\theta),
\end{equation}
where
\begin{equation}\label{eq:FI}
\begin{split}
    I_i(\theta)
    &= \mathbb{E}\!\left[
        \nabla_\theta P_\theta(X^{i} \mid X^{(<i)}) \,
        \nabla_\theta P_\theta(X^{i} \mid X^{(<i)})^\top
    \right] \\
    &= \mathbb{E}\!\left[
        \nabla_\theta P(X^{i} \mid \rho^{\theta}, E^{i}) \,
        \nabla_\theta P(X^{i} \mid \rho^{\theta}, E^{i})^\top
    \right].
\end{split}
\end{equation}

\subsection{Sample Complexity Lower Bound}

In this section, we derive the sample complexity lower bound for the IOMS protocol using Fisher information. The proof consists of three parts. First, we derive the form of the Fisher information matrix for both non-concatenating and concatenating protocols. Then, we prove that the trace of the Fisher information matrix is upper bounded around a local regime for the concatenating protocol. Finally, we apply the Bayesian Cramér--Rao bound to establish the sample complexity lower bound.

\begin{lemma}\label{lm:FImatrix}
For a non-concatenating IOMS protocol, the Fisher information matrix for a single measurement defined in Eq.~\ref{eq:FI} is given by
\begin{equation}\label{eq:FIM}
    I(\theta) = P_s K^\top D(p)^{-1} K P_s,
\end{equation}
where $K_{ij} = \operatorname{Tr}(E_i\rho_j)$, $\rho_j = U_j \rho U_j^\dagger$, $D(p)$ is a diagonal matrix with entries $D(p)_{ii} = p_i = \sum_{j=1}^r \operatorname{Tr}(E_i\rho_j)\theta_j$, and $P_s = I - u u^\top/d$ is the projector enforcing $\sum_i \theta_i = 1$.
\end{lemma}

\begin{proof}
For the non-concatenating IOMS protocol, we prepare a state, apply the mixed unitary channel, and measure the output state. The intermediate unitaries can be absorbed into the preparation and measurement steps, so without loss of generality we can ignore them. By the Born rule, the measurement distribution is
\begin{equation}\label{eq:kthetap}
\begin{split}
    P(X = i \mid \rho^\theta, E)
    &= \operatorname{Tr}(E_i \rho^\theta) \\
    &= \operatorname{Tr}\!\left(E_i \sum_{j=1}^r \theta_j U_j \rho U_j^\dagger \right) \\
    &= \sum_{j=1}^r \theta_j \operatorname{Tr}(E_i \rho_j) \\
    &= \sum_{j=1}^r \theta_j K_{ij}.
\end{split}
\end{equation}
Denote $p_i(\theta) = \sum_j \theta_j K_{ij}$ and let $D(p)$ be the diagonal matrix with entries $p_i(\theta)$.

The Fisher information matrix for a single measurement is defined as
\begin{equation}
    I(\theta)_{ab}
    = \sum_i p_i(\theta)
    \frac{\partial \log p_i(\theta)}{\partial \theta_a}
    \frac{\partial \log p_i(\theta)}{\partial \theta_b}.
\end{equation}
Since
\begin{equation}
    \frac{\partial p_i(\theta)}{\partial \theta_a} = K_{ia},
    \qquad
    \frac{\partial \log p_i(\theta)}{\partial \theta_a}
    = \frac{K_{ia}}{p_i(\theta)},
\end{equation}
we have
\begin{equation}
    I(\theta)_{ab} = \sum_i \frac{K_{ia} K_{ib}}{p_i(\theta)}.
\end{equation}
In matrix form, this can be written as
\begin{equation}
    I_{\mathrm{full}}(\theta) = K^\top D(p)^{-1} K.
\end{equation}

However, $\theta$ is a probability vector satisfying $\sum_{j=1}^r \theta_j = 1$, meaning that the true parameter space is the $(r-1)$-dimensional simplex. Its tangent space at any interior point is
\begin{equation}
    T_\theta = \{ v \in \mathbb{R}^r \mid \mathbf{1}^\top v = 0 \}.
\end{equation}
Let $P_s = I - u u^\top/d$ be the orthogonal projector onto this subspace. The Fisher information restricted to the simplex is therefore
\begin{equation}
    I(\theta) = P_s I_{\mathrm{full}}(\theta) P_s
    = P_s K^\top D(p)^{-1} K P_s.
\end{equation}
This completes the proof.
\end{proof}

Now we give the form of the Fisher information matrix for the concatenating IOMS protocol. Notice that if we concatenate $k$ mixed unitary channels together, with intermediate unitary operations, the resulting overall channel is equivalent to a mixed unitary channel of rank $r^k$. We define this effective mixed unitary channel as follows.

\begin{definition}[Effective mixed unitary channel of $k$-fold concatenation]\label{def:effectk}
Consider a protocol that concatenates $k$ mixed unitary channels, each of the form
\begin{equation}
    \Lambda(\rho) = \sum_{a=1}^r \theta_a U_a \rho U_a^\dagger,
\end{equation}
with intermediate unitaries $V_1,\dots,V_k$ applied between successive uses of $\Lambda$.
The resulting overall channel is itself a mixed unitary channel of rank $r^k$, given by
\begin{equation}
    \tilde{\Lambda}(\cdot)
    = \sum_{i=1}^{r^k} \tilde{\theta}_i\, \tilde{U}_i (\cdot) \tilde{U}_i^\dagger,
\end{equation}
where the index $i$ corresponds to a $k$-tuple $(a_1,\dots,a_k) \in [r]^k$, 
the coefficient vector satisfies $\tilde{\theta} = \theta^{\otimes k}$, and the effective unitary operators are
\begin{equation}
    \tilde{U}_{a_1,\dots,a_k}
    = V_k U_{a_k} V_{k-1} U_{a_{k-1}} \cdots V_2 U_{a_2} V_1 U_{a_1},
\end{equation}
with $V_1,\dots,V_k$ denoting the intermediate unitaries in the concatenating IOMS protocol.
\end{definition}

We now express the Fisher information for the concatenating IOMS protocol in terms of the Fisher information of the effective mixed unitary channel defined above.

\begin{lemma}\label{lm:FI-k-concat}
For a $k$-fold concatenating IOMS protocol, let $\tilde{I}(\tilde{\theta})$ denote the Fisher information matrix for a single measurement of the effective mixed unitary channel $\tilde{\Lambda}$ in Definition~\ref{def:effectk}, viewed as a function of the parameter $\tilde{\theta} \in \mathbb{R}^{r^k}$. Then the Fisher information matrix with respect to the original parameter $\theta \in \mathbb{R}^r$ is
\begin{equation}\label{eq:concFI}
    I^{(k)}(\theta)
    = J(\theta)^{\top} \, \tilde{I}(\tilde{\theta}) \, J(\theta),
\end{equation}
where $\tilde{\theta} = \theta^{\otimes k}$ and
\begin{equation}
    J(\theta) = \frac{\partial \tilde{\theta}}{\partial \theta}
\end{equation}
is the Jacobian of the tensor-product map $\theta \mapsto \theta^{\otimes k}$.
\end{lemma}

\begin{proof}
For a fixed POVM and protocol, the outcome distribution of a single measurement for the concatenating IOMS protocol depends on $\theta$ only through the effective parameter $\tilde{\theta} = \theta^{\otimes k}$. Denote this distribution by
\begin{equation}
    p(x \mid \theta) = p\bigl(x \mid \tilde{\theta}(\theta)\bigr).
\end{equation}
Let $\tilde{I}(\tilde{\theta})$ be the Fisher information matrix with respect to $\tilde{\theta}$, i.e.,
\begin{equation}
    \tilde{I}(\tilde{\theta})
    = \mathbb{E}_{x \sim p(\cdot \mid \tilde{\theta})}
    \left[
        \nabla_{\tilde{\theta}} \log p(x \mid \tilde{\theta})\,
        \nabla_{\tilde{\theta}} \log p(x \mid \tilde{\theta})^{\top}
    \right].
\end{equation}
By the chain rule,
\begin{equation}
    \nabla_{\theta} \log p(x \mid \theta)
    = J(\theta)^{\top} \nabla_{\tilde{\theta}} \log p(x \mid \tilde{\theta}),
\end{equation}
where $J(\theta) = \partial \tilde{\theta} / \partial \theta$ is the Jacobian of the map $\theta \mapsto \theta^{\otimes k}$. Therefore,
\begin{equation}
\begin{split}
    I^{(k)}(\theta)
    &= \mathbb{E}_{x \sim p(\cdot \mid \theta)}
    \left[
        \nabla_{\theta} \log p(x \mid \theta)\,
        \nabla_{\theta} \log p(x \mid \theta)^{\top}
    \right] \\
    &= \mathbb{E}_{x \sim p(\cdot \mid \tilde{\theta})}
    \left[
        J(\theta)^{\top} \nabla_{\tilde{\theta}} \log p(x \mid \tilde{\theta})\,
        \nabla_{\tilde{\theta}} \log p(x \mid \tilde{\theta})^{\top} J(\theta)
    \right] \\
    &= J(\theta)^{\top} \tilde{I}(\tilde{\theta}) J(\theta),
\end{split}
\end{equation}
which proves the claim.
\end{proof}

Before bounding the trace of the Fisher information, we need one auxiliary lemma, which bounds the sum of the maximum entries of $K$ per row.

\begin{lemma}\label{lm:suc}
For the $K$ matrix in Eq.~\ref{eq:FIM}, we have
\begin{equation}
    \sum_{i=1}^s \max_j \{K_{ij}\} \leq d.
\end{equation}
\end{lemma}

\begin{proof}
\begin{equation}
    \sum_{i=1}^s \max_j \{K_{ij}\} 
    = \sum_{i=1}^s \max_j{\operatorname{Tr}(E_i\rho_j)} 
    \leq \sum_{i=1}^s \operatorname{Tr}(E_i) 
    = \operatorname{Tr}(I^d) = d.
\end{equation}
\end{proof}

It is worth emphasizing the intuition behind why the row sum of the $K$ matrix leads to the bound on Fisher information. The row sum of $K$ corresponds to the maximum success probability of identifying \emph{which unitary has been applied}, which represents the total information that can be extracted from a single sample. The inability to perfectly identify the applied unitary is precisely what leads to information loss.

Now we are ready to bound the trace of the Fisher information matrix of a single measurement for both non-concatenating and concatenating protocols at the local parameter corresponding to the uniform distribution.

\begin{theorem}\label{th:trFI}
For the non-concatenating protocol, we have
\begin{equation}\label{eq:trInonc}
    \operatorname{Tr}(I(u)) \le r d.
\end{equation}
For the concatenating protocol with at most $k$ concatenations, we have
\begin{equation}\label{eq:trIc}
    \operatorname{Tr}(I(u)) \le k^2 r d.
\end{equation}
\end{theorem}

\begin{proof}
We first prove Eq.~\ref{eq:trInonc}. From Eq.~\ref{eq:FIM}, we have
\begin{equation}
\begin{split}
    \operatorname{Tr}I(u)
    &= \operatorname{Tr}(P_s K^\top D(p)^{-1} K P_s) \\
    &\le \operatorname{Tr}(K^\top D(p)^{-1} K) \\
    &= \sum_{t=1}^{s}\frac{\sum_{j=1}^{r} K_{tj}^{2}}{p_t} \\
    &\le \sum_{t=1}^{s}\frac{(\max_{1\le j\le r} K_{tj})\sum_{j=1}^{r} K_{tj}}{p_t} \\
    &= \sum_{t=1}^{s}\frac{(\max_{1\le j\le r} K_{tj})\sum_{j=1}^{r} K_{tj}}{[Ku]_t} \\
    &= \sum_{t=1}^{s}\frac{(\max_{1\le j\le r} K_{tj})\sum_{j=1}^{r} K_{tj}}{\frac{1}{r}\sum_{j=1}^{r} K_{tj}} \\
    &= r\sum_{t=1}^{s}\max_{1\le j\le r} K_{tj}\\
    &\le r\,d.
\end{split}
\end{equation}
Where the final inequality uses Lemma~\ref{lm:suc}. Thus, the bound in Eq.~\ref{eq:trInonc} is established.

Now we prove Eq.~\ref{eq:trIc}. From Eq.~\ref{eq:concFI}, we have
\begin{equation}\label{eq:trIJconc}
    \operatorname{Tr}(I^{(k)}(\theta)) 
    = \operatorname{Tr}(J(\theta)^{\top} \, \tilde{I}(\tilde{\theta}) \, J(\theta)) 
    \le \operatorname{Tr}(\tilde{I}(\tilde{\theta})) \| J(\theta) \|_2^2 
    \le r^k d  \| J(\theta) \|_2^2,
\end{equation}
where the second inequality follows because $\tilde{I}(\tilde{\theta})$ is the expectation of rank-one operators, and the last inequality because the effective mixed unitary channel has rank $r^k$.

Next, we upper bound $\| J(\theta) \|_2^2$.
For the tensor map $f(\theta) = \theta^{\otimes k}$ with $\theta \in \mathbb{R}^r$, 
the Jacobian at the uniform point $\theta_* = \mathbf{1}/r$ has entries
\begin{equation}
    J_{ab} = m(a,b)\,r^{1-k},
\end{equation}
where $m(a,b)$ counts how many times index $b$ appears in the $k$-tuple $a \in [r]^k$.
A direct combinatorial calculation gives
\begin{equation}
    J^\top J = k\,r^{-k}\!\left(I_r + (k-1)\mathbf{1}\mathbf{1}^\top\right),
\end{equation}
whose largest eigenvalue is
\begin{equation}
    \lambda_{\max} = k\,r^{-k}\!\left(1 + (k-1)r\right).
\end{equation}
Hence, the operator two-norm of $J$ is
\begin{equation}
    \|J(\theta_*)\|_2
    = \sqrt{k\,r^{-k}\!\left(1 + (k-1)r\right)} \le k r^{(1-k)/2}.
\end{equation}
Substituting this into Eq.~\ref{eq:trIJconc}, we obtain
\begin{equation}
    \operatorname{Tr}(I^{(k)}(\theta_*)) \le r^k d  \| J(\theta) \|_2^2 \le k^2 r d,
\end{equation}
which completes the proof of Eq.~\ref{eq:trIc}.
\end{proof}

Now that we have an upper bound on the trace of the Fisher information matrix, we are ready to bound the sample complexity. The well-known Cramér--Rao bound applies only to unbiased estimators. Van Trees' inequality, on the other hand, applies to general (possibly biased) estimators. Informally, Van Trees' inequality states that the mean-square error (MSE) is lower bounded by the inverse of the sum of the expected Fisher information in a region and the information contributed by the prior distribution on that region. The following lemma shows that, in an asymptotic sense, if the Fisher information is continuous, then the presence of a local parameter with small Fisher information implies a sample-complexity scaling for biased estimators that matches the one suggested by the Cramér--Rao bound for unbiased estimators.

\begin{lemma}[Sample complexity lower bound via Van Trees' inequality]\label{lm:vt}
Let $\{\mathsf{P}_\theta : \theta \in \Theta \subset \mathbb{R}^r\}$ be a regular model with independent single-sample Fisher information matrix $I_1(\theta)$, assumed continuous in $\theta$. 
If an estimator $\theta_N$ achieves 
\begin{equation}
    \sup_{\theta\in\Theta}\mathbb{E}_\theta\!\left[\|\theta_N-\theta\|^2\right] \le \varepsilon^2,
\end{equation}
for all $\theta$, then there exists $\theta_0$ such that 
\begin{equation}
    N = \Omega\!\left(\frac{r^2}{\operatorname{Tr}(I_1(\theta_0))\,\varepsilon^2}\right).
\end{equation}
\end{lemma}

\begin{proof}
By continuity of $I_1(\theta)$, there exists $\theta_0$ and a neighborhood 
$U$ such that $\operatorname{Tr}(I_1(\theta)) \le 2J_0$ for all $\theta\in U$, 
with $J_0 = \operatorname{Tr}(I_1(\theta_0))$.

Choose a smooth prior $\pi$ supported in $U$. 
The Van Trees inequality gives
\begin{equation}
\begin{split}
    R_\pi(\theta_N)
    &= \mathbb{E}_\pi\mathbb{E}_\theta[\|\theta_N-\theta\|^2] \\
    &\ge 
    \operatorname{Tr}\!\left(\big(N\,\mathbb{E}_\pi[I_1(\theta)]+I(\pi)\big)^{-1}\right) \\
    &\ge \frac{r^2}{2N J_0+\operatorname{Tr}(I(\pi))},
\end{split}
\end{equation}
where the last inequality uses 
$\operatorname{Tr}(A^{-1}) \ge r^2/\operatorname{Tr}(A)$ for $A\succ0$.

Since the uniform MSE dominates the Bayes risk, 
$\varepsilon^2 \ge R_\pi(\theta_N)$ implies 
\begin{equation}
    N \ge c\,\frac{r^2}{J_0\,\varepsilon^2}
\end{equation}
for some constant $c>0$. 
\end{proof}

Now we are ready to prove our first main result: the sample complexity lower bound for the IOMS protocol.

\begin{theorem}
    For IOMS protocols with i.i.d.\ measurement outcomes, in order to learn $\theta$ of a mixed unitary channel such that
    \begin{equation}
        \sup_{\theta\in\Theta}\mathbb{E}_\theta\!\left[\|\theta_N-\theta\|^2\right] \le \varepsilon^2,
    \end{equation}
    the number of samples $N_1$ must satisfy
    \begin{equation}
        N_1 \ge \Omega\!\left(\frac{r}{k d \varepsilon^2}\right),
    \end{equation}
    where $r$ is the rank of the mixed unitary channel, $d$ is the dimension of the input probe state, and $k$ is the number of concatenations allowed.
\end{theorem}

\begin{proof}
    Substituting Theorem~\ref{th:trFI} into Lemma~\ref{lm:vt} gives the result. The dependence on $k$ instead of $k^2$ is because each measurement on a $k$-concatenating channel counts as $k$ samples.
\end{proof}

\subsection{Learning Random Mixed Unitary Channels}

The learning difficulty of a mixed unitary channel depends on three factors: its rank, the system dimension, and the specific choice of unitaries. In the previous sections, we showed that the sample complexity lower bound scales with $\frac{r}{d_{\Lambda}}$, which implies that for a mixed unitary channel to be easy to learn, its rank must scale comparably to the system dimension. However, we have not yet discussed the effect of the choice of unitaries on learnability, which is of practical interest—specifically, identifying which choices of unitaries make the channel easier to learn.

Perhaps surprisingly, we show that in high dimensions, randomly selected unitaries are quite easy to learn using a non-concatenating IOMS protocol with ancilla for $r \le d_{\Lambda}^2$, which already corresponds to the maximum degrees of freedom for a general channel. The intuition is that in high-dimensional spaces, the Hilbert space is extremely large, so the states generated by applying Haar-random unitaries to the maximally entangled state are almost mutually orthogonal with overwhelmingly high probability. Hence, one can effectively distinguish which unitary was applied, making Lemma~\ref{lm:suc} nearly tight, and consequently leading to minimal information loss. Notably, this can be achieved using the so-called \emph{Pretty Good Measurement} (PGM)~\cite{iten2016pgm} to distinguish the applied unitaries. We formalize this in the following lemma.

\begin{lemma}\label{lm:pgmkii}
Define the quantum ensemble $\mathcal{E} = \{\rho_i\}_{i=1}^r$, where
\begin{equation}
    \rho_i = (U_i \otimes \mathbb{I}) 
    \lvert \psi \rangle \langle \psi \rvert 
    (U_i^\dagger \otimes \mathbb{I}), 
    \qquad 
    \lvert \psi \rangle
    = \frac{1}{\sqrt{d_\Lambda}}
      \sum_{j=1}^{d_\Lambda} 
      \lvert j \rangle \otimes \lvert j \rangle .
\end{equation}
The Pretty Good Measurement (PGM) associated with the uniform mixture of the ensemble is defined by the POVM elements
\begin{equation}
    E_i 
    = \sigma^{-\frac{1}{2}}
      \left( \frac{1}{r} \rho_i \right)
      \sigma^{-\frac{1}{2}},
    \qquad i = 1,\dots,r,
\end{equation}
where $\sigma = \frac{1}{r}\sum_i \rho_i$. Then, for any $r \le d_{\Lambda}^2$, we have
\begin{equation}\label{eq:kii}
    \lim_{r\rightarrow\infty} 
    \operatorname{Pr}\!\left(
        \forall i,\; \operatorname{Tr}(E_i\rho_i) \ge 0.7
    \right) = 1.
\end{equation}
\end{lemma}

\begin{proof}
Our proof builds upon two known results. First, we use Theorem~4.5 of~\cite{montanaro2007distinguishability} to show that the expectation value $\mathbb{E}[\operatorname{Tr}(E_i\rho_i)]$, which equals the average success probability of the PGM, satisfies $\mathbb{E}[P^{\mathrm{PGM}}(\mathcal{E})] \ge 0.72$. Second, we apply Theorem~5.5 of~\cite{nayak2023rigidity} to show that each $\operatorname{Tr}(E_i\rho_i)$ is highly concentrated around its mean. Combining these two steps and applying a union bound gives the desired result.

From Theorem~4.5 in~\cite{montanaro2007distinguishability}, we have
\begin{equation}
    \mathbb{E}[\operatorname{Tr}(E_i\rho_i)] 
    = \mathbb{E}\!\left[\frac{1}{r}\sum_{i=1}^r \operatorname{Tr}(E_i\rho_i)\right]
    = \mathbb{E}[P^{\mathrm{PGM}}(\mathcal{E})] 
    \ge 0.72.
\end{equation}

Next, we prove that $\operatorname{Tr}(E_i\rho_i)$ is highly concentrated around its mean. 
Using Theorem~5.5 in~\cite{nayak2023rigidity}, which states that any $k$-Lipschitz function of Haar-random unitaries is tightly concentrated around its mean, we first show that $\operatorname{Tr}(E_i\rho_i)$ is $2$-Lipschitz with respect to the unitaries:
\begin{equation}
\begin{split}
    |f(U_1,\dots,U_r) - f(U_1',\dots,U_r')| 
    &= |\operatorname{Tr}(E_i\rho_i) - \operatorname{Tr}(E_i'\rho_i')| \\
    &\leq \|E_i\|\, \|\rho_i - \rho_i'\|_1 \\
    &\leq \|\rho_i - \rho_i'\|_1 \\
    &\leq 2 \|\,|\psi_i\rangle - |\psi_i'\rangle\,\| \\
    &= 2\sqrt{\langle \psi | (U_i - U_i')^\dagger (U_i - U_i') | \psi \rangle} \\
    &\leq 2\|U_i - U_i'\|.
\end{split}
\end{equation}
Thus, the function is 2-Lipschitz. 
By Theorem~5.5 in~\cite{nayak2023rigidity}, we then have
\begin{equation}
\begin{split}
    \Pr\!\left( \operatorname{Tr}(E_i \rho_i) 
    \le \mathbb{E}\bigl[\operatorname{Tr}(E_i \rho_i)\bigr] - t \right)
    &\le \exp\!\left( -\frac{(d - 1)t^2}{96} \right) \\
    &\le \exp\!\left( -\frac{(\sqrt{r} - 1)t^2}{96} \right).
\end{split}
\end{equation}

This shows that $\operatorname{Tr}(E_i\rho_i)$ is sharply concentrated around its mean. 
Using the union bound, we obtain
\begin{equation}\label{eq:union}
\begin{split}
    \operatorname{Pr}\!\left(\forall i,\; \operatorname{Tr}(E_i\rho_i) \ge 0.7\right) 
    &\ge 1 - r\, \operatorname{Pr}\!\left(\operatorname{Tr}(E_i\rho_i) < 0.7\right) \\
    &\ge 1 - r\, \exp\!\left(-\frac{(\sqrt{r}-1)\,0.02^2}{96}\right),
\end{split}
\end{equation}
and the right-hand side approaches $1$ as $r \to \infty$, proving Eq.~\ref{eq:kii}.
\end{proof}

The learning algorithm we use is simple: prepare a maximally entangled probe state, apply the mixed unitary channel to one subsystem, and perform the PGM to identify which unitary was applied. Repeating this procedure produces i.i.d.\ measurement outcomes whose distribution depends linearly on the unknown parameter vector $\theta$. Consequently, estimating this distribution and solving a linear system yields an estimator for $\theta$. The detailed procedure is given below.

\begin{algorithm}[H]
\caption{PGM-based estimator for a mixed unitary channel}
\label{alg:pgm-mixed-unitary}
\begin{algorithmic}[1]
\Require Known unitaries $\{U_i\}_{i=1}^r \subset \mathbb{C}^{d_\Lambda \times d_\Lambda}$ and access to
$\Lambda(\rho) = \sum_{i=1}^r \theta_i U_i \rho U_i^\dagger$ with $\sum_{i=1}^r \theta_i = 1$; number of channel uses $N$.
\Ensure Estimate $\tilde{\theta} \in \mathbb{R}^r$.
\State Prepare the maximally entangled state
$|\psi\rangle = d_\Lambda^{-1/2} \sum_{j=1}^{d_\Lambda} |j\rangle \otimes |j\rangle$.
\State For each $j$, set $\rho_j = (U_j \otimes \mathbb{I})|\psi\rangle\langle\psi|(U_j^\dagger \otimes \mathbb{I})$ and
$\sigma = \frac{1}{r}\sum_{j=1}^r \rho_j$.
\State Define the PGM elements $E_i = \sigma^{-1/2}(\rho_i/r)\sigma^{-1/2}$ and overlaps
$K_{ij} = \operatorname{Tr}(E_i \rho_j)$; compute $K^{-1}$ (or pseudoinverse $K^+$).
\State Initialize counts $N_i \gets 0$ for $i = 1,\dots,r$.
\For{$t = 1$ \textbf{to} $N$}
    \State Prepare $|\psi\rangle$ and apply $(\Lambda \otimes \operatorname{id})(|\psi\rangle\langle\psi|)$.
    \State Measure with POVM $\{E_i\}_{i=1}^r$, obtain outcome $x_t$, and update $N_{x_t} \gets N_{x_t} + 1$.
\EndFor
\State Set empirical frequencies $\hat{p}_i = N_i/N$ for $i = 1,\dots,r$.
\State Estimate $\tilde{\theta} = K^{-1}\hat{p}$ (or $\tilde{\theta} = K^{+}\hat{p}$ if $K$ is not invertible).
\State \Return $\tilde{\theta}$.
\end{algorithmic}
\end{algorithm}

We now show that this estimator achieves the optimal sample-complexity scaling for random mixed unitary channels.

\begin{theorem}
In high dimensions, for a $d_{\Lambda}$-dimensional mixed unitary channel with rank $r = d_{\Lambda}^2$, where the unitaries are independently drawn from the Haar measure, Algorithm~\ref{alg:pgm-mixed-unitary} learns $\theta$ such that $\mathbb{E}[\|\tilde{\theta} -\theta \|^2] \le \varepsilon^2$ using
\begin{equation}
    N = \mathcal{O}\!\left(\frac{1}{\varepsilon^2}\right)
\end{equation}
samples.
\end{theorem}

\begin{proof}
For the IOMS protocol, the probability distribution for each measurement outcome is given by Eq.~\ref{eq:kthetap}, i.e., $p = K\theta$, where $K_{ij} = \operatorname{Tr}(E_i \rho_j)$. 
Since we use the PGM and a maximally entangled probe state, Lemma~\ref{lm:pgmkii} implies that $\forall i,\; K_{ii} > 0.7$ with overwhelmingly high probability. 
Under this condition, because $K$ is symmetric with row and column sums equal to $1$, Gerschgorin’s theorem gives
\begin{equation}
    \lambda_{\min}(K) \ge 0.7 - (1 - 0.7) = 0.4.
\end{equation}

The measurements are i.i.d.\ samples from $p = K\theta$. 
By estimating $p$ and solving $\theta = K^{-1}\hat{p}$, we obtain
\begin{equation}
    \mathbb{E}\!\left[\|\hat{p} - p\|_2^2\right] 
    = \frac{1}{N}\!\left(\sum_{i=1}^d p_i(1 - p_i)\right)
    \le \frac{1}{N}.
\end{equation}
Consequently,
\begin{equation}
    \mathbb{E}\!\left[\|\tilde{\theta} - \theta\|_2^2\right]
    = \mathbb{E}\!\left[\|K^{-1}(\hat{p} - p)\|_2^2\right]
    \le \|K^{-1}\|_2^2 \, \mathbb{E}\!\left[\|\hat{p} - p\|_2^2\right]
    \le \frac{\|K^{-1}\|_2^2}{N}.
\end{equation}
Using $\lambda_{\min}(K)\ge 0.4$, we have $\|K^{-1}\|_2 \le 1/0.4 = 2.5$, giving
\begin{equation}
    \mathbb{E}\!\left[\|\tilde{\theta} - \theta\|_2^2\right]
    \le \frac{(2.5)^2}{N}
    = \frac{6.25}{N}.
\end{equation}
Thus, the MSE of parameter estimation scales as $\mathcal{O}(1/N)$ with an explicit constant bounded by $6.25$. Therefore,
\begin{equation}
    \mathbb{E}\!\left[\|\tilde{\theta} - \theta\|_2^2\right] \le \varepsilon^2
    \quad \text{whenever} \quad
    N \ge \frac{6.25}{\varepsilon^2}.
\end{equation}
Hence, there exists a universal constant $C \le 6.25$ such that $N \ge C/\varepsilon^2$, completing the proof.
\end{proof}

\section{Conclusions}

We studied the learnability of mixed unitary channels and formally established an \emph{exponential separation} in sample complexity between protocols that may employ entanglement across channel uses and those that do not. Our theorems recover prior results on Pauli channels as a special case, but are formulated in terms of the mean\mbox{-}squared error (MSE), rather than the more commonly used $\varepsilon$--$\delta$ formulation in the quantum learning literature. There are ongoing research aims to extend our analysis to $\varepsilon$--$\delta$ type performance guarantees via Fisher information.

Within our general framework, we show that ancilla systems are necessary for achieving optimal efficiency only when the channel rank exceeds the system dimension; otherwise, optimal (or near\mbox{-}optimal) learning can be achieved without ancilla, with success probability approaching one as $d \to \infty$. Furthermore, we demonstrate that random mixed unitary channels are particularly easy to learn by explicitly constructing an algorithm that attains the sample\mbox{-}complexity lower bound, thereby showing that our bound is asymptotically tight.

\textit{Practical directions.} Mixed unitary channels offer a flexible error model for quantum error mitigation beyond Pauli noise. Because Pauli channels typically require entanglement to reach optimal sample complexity---and because the number of Pauli terms grows exponentially with the number of qubits---Pauli\mbox{-}only models can be both resource\mbox{-}intensive and brittle under truncation. By contrast, mixed unitary channel models can capture structured, non\mbox{-}Pauli noise with fewer mixture components, and in many instances admit entanglement\mbox{-}free (or lower\mbox{-}entanglement) learning protocols that are provably optimal. Exploring mixed unitary\mbox{-}based mitigation pipelines on near\mbox{-}term devices is a promising avenue. 

\textit{Theoretical directions.} A natural next step is to extend our Fisher\mbox{-}information--based analysis to broader classes of quantum channels. While general Fisher metrics for quantum channels are known in full generality~\cite{yuan2017fidelity}, their breadth can obscure the structural insights that guide algorithm design. Our goal is to identify intermediate\mbox{-}complexity families---richer than mixed unitary channels yet still structured---for which one can derive \emph{simple, informative bounds} that (i) pinpoint what makes a channel easy or hard to learn (e.g., rank, dimension, symmetry, or commutativity), and (ii) come with constructive protocols that saturate these bounds. Additional open problems include characterizing finite\mbox{-}sample (non\mbox{-}asymptotic) rates, sharpening resource trade\mbox{-}offs (entanglement, ancilla dimension, circuit depth, and classical post\mbox{-}processing).

\textit{Outlook.} Overall, our results provide a clean criterion for when entanglement is truly needed and establish baseline\mbox{-}optimal procedures for random mixed unitary channel learning, opening a path toward principled error models and scalable learning protocols in more general settings.

\section*{Acknowledgements}

We thank Senrui Chen for helpful discussions on the differences between mixed unitary and Pauli channels, and Hyukgun Kwon for introducing the Fisher-information approach to sample-complexity bounds in quantum learning. We acknowledge support from the ARO(W911NF-23-1-0077), ARO MURI (W911NF-21-1-0325), AFOSR MURI (FA9550-21-1-0209, FA9550-23-1-0338), DARPA (HR0011-24-9-0359, HR0011-24-9-0361), NSF (ERC-1941583, OMA-2137642, OSI-2326767, CCF-2312755, OSI-2426975), and the Packard Foundation (2020-71479).


\begin{thebibliography}{99}

\bibitem{chen2022quantum}
Chen, Senrui and Zhou, Sisi and Seif, Alireza and Jiang, Liang.
Quantum advantages for Pauli channel estimation.
\emph{Physical Review A} \textbf{105}, 032435 (2022).

\bibitem{flammia2020efficient}
Flammia, Steven T., and Wallman, Joel J.
Efficient estimation of Pauli channels.
\emph{ACM Transactions on Quantum Computing} \textbf{1}(1), 1--32 (2020).

\bibitem{chen2024tight}
Chen, Senrui, et al.
Tight bounds on Pauli channel learning without entanglement.
\emph{Physical Review Letters} \textbf{132}(18), 180805 (2024).

\bibitem{cai2023qem}
Cai, Zhenyu, et al.
Quantum error mitigation.
\emph{Reviews of Modern Physics} \textbf{95}(4), 045005 (2023).

\bibitem{yuan2017fidelity}
Yuan, Haidong and Fung, Chi-Hang Fred.
Fidelity and Fisher information on quantum channels.
\emph{New Journal of Physics} \textbf{19}, 113039 (2017).

\bibitem{Hausladen01121994}
Hausladen, Paul and Wootters, William K.
A `Pretty Good' Measurement for Distinguishing Quantum States.
\emph{Journal of Modern Optics} \textbf{41}, 2385--2390 (1994).

\bibitem{Gill1995vanTrees}
Gill, Richard D. and Levit, Boris Y.
Applications of the van Trees inequality: a Bayesian Cram\'er--Rao bound.
\emph{Bernoulli} \textbf{1} (1995).

\bibitem{montanaro2007distinguishability}
Montanaro, Ashley.
On the distinguishability of random quantum states.
\emph{Communications in Mathematical Physics} \textbf{273}, 619--636 (2007).

\bibitem{nayak2023rigidity}
Nayak, Ashwin and Yuen, Henry.
Rigidity of superdense coding.
\emph{ACM Transactions on Quantum Computing} \textbf{4}, 1--39 (2023).

\bibitem{iten2016pgm}
Iten, Raban, Renes, Joseph M., and Sutter, David.
Pretty good measures in quantum information theory.
\emph{IEEE Transactions on Information Theory} \textbf{63}(2), 1270--1279 (2016).

\end{thebibliography}
\end{document}